%% file: mmdl.tex
\documentclass[a4paper,10pt]{llncs}

\usepackage{\jobname}

\usepackage{bussproofs}

\title{Axiomatizing Propositional Dependence Logics}

\author{Katsuhiko Sano\inst{1} \and Jonni~Virtema\inst{1,2}}

\institute{
Japan Advanced Institute of Science and Technology, Japan
\and
University of Tampere, Finland\\
\email{\{katsuhiko.sano, jonni.virtema\}@gmail.com}}

\begin{document}

\maketitle

\begin{abstract}
We give sound and complete Hilbert-style axiomatizations for propositional dependence logic ($\PD$), modal dependence logic ($\MDL$), and extended modal dependence logic ($\EMDL$) by extending existing axiomatizations for propositional logic and modal logic. In addition, we give novel labeled tableau calculi for $\PD$, $\MDL$, and $\EMDL$. We prove soundness, completeness and termination for each of the labeled calculi. 
\end{abstract}

\section{Introduction}
Functional dependences occur everywhere in science, e.g., in descriptions of discrete systems, in database theory, social choice theory, mathematics, and physics. Modal logic is an important formalism utilized in the research of numerous disciplines including many of the fields mentioned above. 
With the aim to express functional dependences in the framework of modal logic, V\"a\"an\"anen \cite{va08} introduced \emph{modal dependence logic} ($\MDL$). Modal dependence logic extends modal logic with \emph{propositional dependence atoms}. A dependence atom, denoted by $\dep{p_1,\dots,p_n,q}$, intuitively states that the truth value of the proposition $q$ is functionally determined by the truth values of the propositions $p_1,\dots,p_{n}$. It was soon realized that $\MDL$ lacks the ability to express temporal dependencies; there is no mechanism in $\MDL$ to express dependencies that occur between different points of the model. This is due to the restriction that only proposition symbols are allowed in the dependence atoms of modal dependence logic. To overcome this defect Ebbing et~al. \cite{EHMMVV13} introduced the
\emph{extended modal dependence logic} ($\EMDL$) by extending the scope of dependence atoms to arbitrary modal formulae, i.e., dependence atoms in extended modal dependence logic are of the form $\dep{\varphi_1,\dots\varphi_n,\psi}$, where $\varphi_1,\dots,\varphi_n,\psi$ are formulae of modal logic.

In recent years the research around modal dependence logic has been very active. The focus has been in the computational complexity and expressive power of related formalisms. Hella et~al. \cite{HeLuSaVi14} established that exactly the properties of teams that are downward closed and closed under the so-called team $k$-bisimulation, for some finite $k$, are definable in extended modal dependence logic. This characterization truly demonstrates the naturality of $\EMDL$. For recent research related to computational complexity of modal dependence logics see, e.g., \cite{EHMMVV13,ebloya,KMSV14,lohvo13,vollmer13,virtema14}. E.g., in \cite{virtema14} it was shown that the validity problem for both $\MDL$ and $\EMDL$ is $\NEXPTIME$-hard and contained in $\NEXPTIME^\NP$. Furthermore, it was shown that the corresponding problem for the propositional fragment $\PD$ of $\MDL$ is $\NEXPTIME$-complete (for the definition of $\PD$ see Section \ref{defprop}). 

In this paper we give sound and complete axiomatizations for variants of propositional and modal dependence logics.
We give Hilbert-style axiomatizations for these logics by extending existing axiomatizations for propositional logic and modal logic. In addition, we give novel labeled tableau calculi for these logics. 
This paper is one of the first articles on proof theory of propositional and modal dependence logics. The only other work known by the authors of this article is the PhD thesis of Fan Yang \cite{fanthesis}. Among other things, in her thesis, Yang presents axiomatizations of variants of propositional dependence logic based on natural deduction. The axiomatizations of Yang are however quite complicated and do not cover variants of modal dependence logic.

The article is structured as follows. In Section \ref{preli} we introduce the required notions and definitions. In Section \ref{axiomatization1} we give Hilbert-style axiomatizations for propositional and modal dependence logics. In Section \ref{axiomatization2} we present labeled tableau calculi for these logics.

\section{Preliminaries}\label{preli}

%
The syntax of propositional logic ($\PL$) and modal logic ($\ML$) could be defined in any standard way. However, when we consider 
the extensions of $\PL$ and $\ML$ by dependence atoms, it is useful to assume that all formulas are
in \emph{negation normal form}, i.e., negations occur only in front of atomic propositions.
Thus we will define the syntax of $\PL$ and $\ML$ in negation normal form. When $\varphi$ is a formula of $\PL$ or $\ML$, we denote by $\varphi^\bot$ the equivalent formula that is obtained from $\neg \varphi$ by pushing all negations to the atomic level. Furthermore, we define $\varphi^\top\dfn \varphi$.
When $\vec{a}$ is a tuple of symbols of length $k$, we denote by $a_j$ the $j$th element of $\vec{a}$, $j\leq k$. 
When $\varphi$ is a formula, $\lvert \varphi \rvert$ denotes the number of symbols in $\varphi$ excluding negations and brackets.
When $A$ is a set $\lvert A \rvert$ denotes the number of elements in $A$. When $f:A\to B$ is a function and $C\subseteq A$, we define $f[C] \dfn \{ f(a) \mid a\in C \}$.

\subsection{Propositional logic with team semantics}\label{defprop}
Let 
$\mathsf{PROP}$ = $\{z_i\mid i\in \mathbb{N} \}$
denote the set of exactly all \emph{propositional variables}, i.e., \emph{proposition symbols}. We mainly use metavariables $p,q,p_1,p_2,q_1,q_2$, etc., in order to refer to variable symbols in $\mathsf{PROP}$. 
Let $D$ be a finite, possibly empty, subset of $\mathsf{PROP}$. A function $s:D\to \{0,1\}$ is called an \emph{assignment}. A set $X$ of assignments $s:D\to \{0,1\}$ is called a \emph{propositional team}. The set $D$ is the \emph{domain} of $X$. Note that the empty team $\emptyset$ does not have a unique domain; any subset of $\mathsf{PROP}$ is a domain of the empty team. By $\{0,1\}^D$, we denote the set of all assignments $s:D\to\{0,1\}$.

Let $\Phi$ be a set of proposition symbols. The syntax for propositional logic $\PL(\Phi)$ is defined as follows:
\[
\varphi \ddfn p\mid \neg p \mid (\varphi \wedge \varphi) \mid (\varphi \vee \varphi), \quad \text{where $p\in\Phi$.}
\]
We will now give the team semantics for propositional logic. 
As we will see below, the team semantics and the ordinary semantics for propositional logic defined via assignments, in a rather strong sense, coincide.
\begin{definition}
Let $\Phi$ be a set of atomic propositions and let $X$ be a propositional team. The satisfaction relation $X\models \varphi$ for $\PL(\Phi)$ is defined as follows. Note that, we always assume that the proposition symbols that occur in $\varphi$ are also in the domain of $X$.
\begin{align*}
X\models p  \quad\Leftrightarrow\quad& \forall s\in X: s(p)=1.\\
X\models \neg p \quad\Leftrightarrow\quad& \forall s\in X: s(p)=0.\\
X\models (\varphi\land\psi) \quad\Leftrightarrow\quad& X\models\varphi \text{ and } X\models\psi.\\
X\models (\varphi\lor\psi) \quad\Leftrightarrow\quad& Y\models\varphi \text{ and } 
Z\models\psi,
\text{ for some $Y,Z$ such that $Y\cup Z= X$}.
\end{align*}
\end{definition}
\begin{proposition}[\cite{Sevenster:2009}]\label{PLflat}
Let $\varphi$ be a formula of propositional logic and let $X$ be a propositional team. Let $\models_\mathcal{PL}$ denote the ordinary satisfaction relation of propositional logic defined via assignments. Then 
$X\models \varphi\,\Leftrightarrow\,\forall s\in X: s\models_\mathcal{PL} \varphi$,
and especially
$\{s\}\models \varphi \,\Leftrightarrow\,  s\models_\mathcal{PL} \varphi$.
\end{proposition}
The syntax of \emph{propositional logic with intuitionistic disjunction} $\PL(\idis)$ is obtained by extending the syntax of $\PL(\Phi)$ by the grammar rule
\(
\varphi \ddfn (\varphi\idis\varphi),
\)
whereas the syntax of \emph{propositional dependence logic} $\PD(\Phi)$ is obtained by extending the syntax of $\PL(\Phi)$ by the grammar rule
$\varphi \ddfn \dep{p_1,\dots,p_n,q}$, where $p_1,\dots,p_n,q\in\Phi$.
The intuitive meaning of the \emph{propositional dependence atom} $\dep{p_1,\dots,p_n,q}$ is that the truth value of the proposition symbol $q$ solely depends on the truth values of the proposition symbols $p_1,\dots,p_n$. The semantics for the intuitionistic disjunction and the propositional dependence atom is defined as follows:
\begin{align*}
X\models (\varphi\idis\psi) \quad\Leftrightarrow\quad& X\models \varphi \text{ or } X\models \psi\\
X\models \dep{p_1,\dots,p_n,q} \quad\Leftrightarrow\quad& \forall s,t\in X: s(p_1)=t(p_1), \dots, s(p_n)=t(p_n)\\
&\text{implies that } s(q)=t(q).
\end{align*}
The next proposition is very useful. The proof is very easy and the result is stated, for example, in \cite{fanthesis}.
\begin{proposition}[Downwards closure]\label{dcprop}
Let $\varphi$ be a formula of $\PL(\idis)$ or $\PD$ and let $Y\subseteq X$ be propositional teams. Then 
$X\models \varphi$ implies $Y\models \varphi$.
\end{proposition}
Note that, by downwards closure, $X\models (\varphi\lor\psi)$ iff $Y \models \varphi$ and $X \setminus Y \models \psi$ for some $Y \subseteq X$.

\subsection{Modal logics}
In this article, in order to keep the notation light, we restrict our attention to mono-modal logic, i.e., to modal logic with just two modal operators ($\Diamond$ and $\Box$). However this is not really a restriction, since the definitions, results, and proofs of this article generalize, in a straightforward manner, to handle also the poly-modal case.

Let $\Phi$ be a set of atomic propositions. The set of formulae for \emph{standard mono-modal logic} $\ML(\Phi)$ is generated by the following grammar
\[
\varphi \ddfn p\mid \neg p \mid (\varphi \wedge \varphi) \mid (\varphi \vee \varphi) \mid \Diamond \varphi \mid \Box \varphi, \quad\text{where $p\in\Phi$.}
\]
Note that, since negations are allowed only in front of proposition symbols, $\Box$ and $\Diamond$ are \emph{not} interdefinable.
The syntax of \emph{modal logic with intuitionistic disjunction} $\ML(\varovee)(\Phi)$ is obtained by extending the syntax of $\ML(\Phi)$ by the grammar rule
\(
\varphi \ddfn (\varphi\varovee\varphi).
\)
The \emph{team semantics for modal logic} is defined via \emph{Kripke models} and \emph{teams}. In the context of modal logic, teams are subsets of the domain of the model.
\begin{definition}
Let $\Phi$ be a set of atomic proposition symbols. A \emph{Kripke model} $\mathrm{K}$ over $\Phi$ is a tuple $\mathrm{K} = (W, R, V)$, where $W$ is a nonempty set of \emph{worlds}, $R\subseteq W\times W$ is a binary relation, and $V\colon \Phi \to \mathcal{P}(W)$ is a \emph{valuation}. A subset $T$ of $W$ is called a \emph{team} of $\mathrm{K}$. Furthermore, define that
\[
R[T] := \{w\in W \mid \exists v\in T \text{ s.t. } vRw \}, \,\,\,
R^{-1}[T] := \{w\in W \mid \exists v\in T \text{ s.t. }  wRv\}.
\]
For teams $T,S\subseteq W$, we write $T[R]S$ if $S\subseteq R[T]$ and $T\subseteq R^{-1}[S]$. Thus, $T[R]S$ holds if and only if for every $w\in T$ there exists some $v\in S$ such that $wRv$, and for every $v\in S$ there exists some $w\in T$ such that $wRv$.
\end{definition}
We are now ready to define the team semantics for modal logic and modal logic with intuitionistic disjunction. Similar to the case of propositional logic, the team semantics of modal logic, in a rather strong sense, coincides with the traditional semantics of modal logic defined via pointed Kripke models.
\begin{definition}
Let $\Phi$ be a set of atomic propositions, $\mathrm{K}$ a Kripke model and $T$ a team of $\mathrm{K}$. The satisfaction relation $\mathrm{K},T\models \varphi$ for $\ML(\Phi)$ is defined as follows. 
\begin{align*}
\mathrm{K},T\models p  \quad\Leftrightarrow\quad& w\in V(p) \, \text{ for every $w\in T$.}\\
\mathrm{K},T\models \neg p \quad\Leftrightarrow\quad& w\not\in V(p) \, \text{ for every $w\in T$.}\\
\mathrm{K},T\models (\varphi\land\psi) \quad\Leftrightarrow\quad& \mathrm{K},T\models\varphi \text{ and } K,T\models\psi.\\
\mathrm{K},T\models (\varphi\lor\psi) \quad\Leftrightarrow\quad& \mathrm{K},T_1\models\varphi \text{ and } 
\mathrm{K},T_2\models\psi \, \text{ for some $T_1$ and $T_2$}\\
&\text{such that $T_1\cup T_2= T$}.\\
\mathrm{K},T\models \Diamond\varphi \quad\Leftrightarrow\quad& \mathrm{K},T'\models\varphi \text{ for some $T'$ such that $T[R]T'$}.\\
\mathrm{K},T\models \Box\varphi \quad\Leftrightarrow\quad& \mathrm{K},T'\models\varphi, \text{ where $T'=R[T]$}.\\
\intertext{For $\ML(\varovee)$ we have the following additional clause:}
\mathrm{K},T\models (\varphi\varovee\psi) \quad\Leftrightarrow\quad& \mathrm{K},T\models\varphi \text{ or } \mathrm{K},T\models\psi.
\end{align*}
\end{definition}
\begin{proposition}[\cite{Sevenster:2009}]
Let $\varphi\in\ML$, $\mathrm{K}$ be a Kripke model and $T$ a team of $\mathrm{K}$. Let $\models_{\ML}$ denote the ordinary satisfaction relation of modal logic defined via pointed Kripke models. Then
\(
\mathrm{K},T\models \varphi \,\Leftrightarrow\, \forall w\in T:\mathrm{K},w\models_{\ML} \varphi
\)
and especially
\(
\mathrm{K},\{w\}\models \varphi \,\Leftrightarrow\, \mathrm{K},w\models_{\ML} \varphi.
\)
\end{proposition}

The syntax for \emph{modal dependence logic} $\MDL(\Phi)$ is obtained by extending the syntax of $\ML(\Phi)$ by propositional dependence atoms
\(
\varphi\ddfn\dep{p_1,\dots, p_n,q},
\)
where $p_1,\dots,p_n,q\in \Phi$, whereas the syntax for \emph{extended modal dependence logic} $\EMDL(\Phi)$ is obtained by extending the syntax of $\ML(\Phi)$ by \emph{modal dependence atoms}
\(
\varphi\ddfn\dep{\varphi_1,\dots, \varphi_n,\psi},
\)
where $\varphi_1,\dots,\varphi_n,\psi\in\ML(\Phi)$. 
The intuitive meaning of the modal dependence atom $\dep{\varphi_1,\dots, \varphi_n,\psi}$ is that the truth value of the formula $\psi$ is completely determined by the truth values $\varphi_1,\dots, \varphi_n$.
The semantics for these dependence atoms is defined as follows.
\begin{align*}
\mathrm{K},T\models \dep{\varphi_1,\dots,\varphi_n,\psi} \quad\Leftrightarrow\quad& \forall w,v\in T: \bigwedge_{1 \leq i \leq n}(\mathrm{K},\{w\}\models\varphi_i \Leftrightarrow \mathrm{K},\{v\}\models\varphi_i)\\
& \text{implies }(\mathrm{K},\{w\}\models\psi\Leftrightarrow \mathrm{K},\{v\}\models\psi).
\end{align*}

The following result for $\MDL$ and $\ML(\varovee)$ is due to \cite{va08} and \cite{ebloya}, respectively. For $\EMDL$ it follows via a translation from $\EMDL$ into $\ML(\varovee)$, see \cite{EHMMVV13}.
\begin{proposition}[Downwards closure]\label{dcml}
Let $\varphi$ be a formula of $\ML(\varovee)$ or $\EMDL$, let $\mathrm{K}$ be a Kripke model and let $S\subseteq T$ be teams of $\mathrm{K}$. Then $\mathrm{K},T\models \varphi$ implies $\mathrm{K},S\models \varphi$.
\end{proposition}

\subsection{Equivalence and validity in team semantics}
We say that formulas  $\varphi$ and $\psi$ of $\PL(\idis)(\Phi)$ or $\PD(\Phi)$ are \emph{equivalent} and write $\varphi\equiv \psi$, if the equivalence
\(
X\models \varphi \,\Leftrightarrow\, X\models \psi
\)
holds for every propositional team $X$ of some finite domain $D\subseteq\Phi$. Likewise, we say that formulas  $\varphi$ and $\psi$ of $\ML(\idis)(\Phi)$ or $\EMDL(\Phi)$ are \emph{equivalent} and write $\varphi\equiv \psi$, if the equivalence
\(
K,T\models \varphi \,\Leftrightarrow\, K,T\models \psi
\)
holds for every Kripke model $K$ and team $T$ of $K$.

A formula $\varphi$ of $\PL(\idis)(\Phi)$ or $\PD(\Phi)$ is said to be \emph{valid}, if $X\models\varphi$ holds for all teams $X$ of some finite domain $D\subseteq\Phi$.
Analogously, a formula $\psi$ of $\EMDL(\Phi)$ or $\ML(\varovee)(\Phi)$ is 
said to be \emph{valid}, if $\mathrm{K},T\models\psi$ holds for every Kripke model $\mathrm{K}$ and every team $T$ of $\mathrm{K}$. When $\varphi$ is a valid formula of $\mathcal{L}$, we write $\models_{\mathcal{L}}\varphi$. 

The following proposition shown in \cite{virtema14} will later proof to be very useful. 

\begin{proposition}[$\idis$-disjunction property]\label{disjunctionprop}
Let $\mathcal{L}\in \{\PL(\idis), \ML(\idis)\}$. For every $\varphi,\psi$ in $\mathcal{L}$, 
$\models_{\mathcal{L}}(\varphi\idis\psi)$ iff $\models_{\mathcal{L}}\varphi$ or $\models_{\mathcal{L}}\psi$. 
\end{proposition}

\section{Extending axiomatizations of $\PL$ and $\ML$}\label{axiomatization1}
In this section we show how to extend sound and complete axiomatizations for $\PL$ and $\ML$ into sound and complete axiomatizations for $\PL(\varovee)$ and $\ML(\varovee)$, respectively. We use the fact that both $\PL(\varovee)$ and $\ML(\varovee)$ have the $\idis$-disjunction property. In addition, we obtain axiomatizations for $\PD$, $\MDL$, and $\EMDL$. The axiomatizations are based on compositional translations from $\PD$ into $\PL(\varovee)$, and from $\MDL$ and $\EMDL$ into $\ML(\varovee)$.

\subsection{Axiomatizations for $\PL(\idis)$ and $\ML(\idis)$}
In the definition below, we treat different occurrences of formulas as distinct entities.
\begin{definition}
Let $\varphi$ be a formula of $\PL(\idis)$ or $\ML(\idis)$. Let $\mathrm{SubOcc}(\varphi)$ denote the \emph{set of exactly all occurrences of subformulas of $\varphi$}.  Define
$$\mathrm{SubOcc}_{\idis}(\varphi)\dfn \{(\psi\idis\theta) \mid (\psi\idis\theta)\in\mathrm{SubOcc}(\varphi)\}.$$
We call a function $f:\mathrm{SubOcc}_{\idis}(\varphi)\to \mathrm{SubOcc}(\varphi)$ a \emph{$\idis$-selection function for $\varphi$} if $f\big((\psi\idis\theta)\big)\in \{\psi,\theta\}$, for every $(\psi\idis\theta)\in \mathrm{SubOcc}_{\idis}(\varphi)$. If $f$ is a $\idis$-selection function for $\varphi$, then $\varphi^f$ denotes the formula that is obtained from $\varphi$ by replacing simultaneously each $(\psi\idis\theta) \in  \mathrm{SubOcc}_{\idis}(\varphi)$ by $f(\psi\idis\theta)$.
\end{definition}
Note that if $\varphi\in\PL(\idis)$, $\psi\in\ML(\idis)$, $f$ is a $\idis$-selection function for $\varphi$, and $g$ is a $\idis$-selection function for $\psi$  then $\varphi^f\in\PL$ and $\psi^g\in\ML$.

\begin{proposition}[\cite{virtema14}]\label{normalform}
Let $\varphi$ be a formula of $\PL(\idis)$ or $\ML(\idis)$, and let $F$ be the set of exactly all $\idis$-selection functions for $\varphi$. Then, $\varphi \equiv \Idis_{f\in F} \varphi^f$. 
\end{proposition}

Let $\mathbf{H}_{\PL}$ and $\mathbf{H}_{\ML}$ denote sound and complete axiomatizations of the negation normal form fragments of $\PL$ and $\ML$, respectively. For a logic $\mathcal{L}$, an \emph{$\mathcal{L}$-context} is a formula of the logic $\mathcal{L}$ extended with the grammar rule $\varphi\ddfn *$. By $\varphi(\psi \,/\, *)$ we denote the formula that is obtained form $\varphi$ by uniformly substituting each occurrence of $*$ in $\varphi$ by $\psi$. 
We are now ready to define the axiomatizations for $\PL(\idis)$ and $\ML(\idis)$.
%
%
%
%
We use $\PL(\idis)$- and $\ML(\idis)$-contexts in the following rules:
\begin{center}
\AxiomC{$\varphi(\psi_i \,/\, *)$}
\RightLabel{$(I \idis i)$}
\UnaryInfC{$\varphi\big((\psi_1\idis\psi_2) \,/\, * \big)$}
\DisplayProof
\quad
$i\in\{1,2\}$.
\end{center}
Let $\mathbf{H}_{\PL(\idis)}$ (or, $\mathbf{H}_{\ML(\idis)}$) be the calculus $\mathbf{H}_{\PL}$ (or, $\mathbf{H}_{\ML}$, respectively) extended with the rules $(I \idis 1)$ and $(I \idis 2)$.
\begin{theorem}\label{easycalculi}
$\mathbf{H}_{\PL(\idis)}$ and
$\mathbf{H}_{\ML(\idis)}$ are sound and complete.
\end{theorem}
\begin{proof}
We will proof the soundness and completeness for $\mathbf{H}_{\PL(\idis)}$. The case for $\mathbf{H}_{\ML(\idis)}$ is completely analogous.

For soundness, it suffices to show that the rule $(I \idis 1)$ preserves validity. The case for $(I \idis 2)$ is symmetric. Let $\varphi$ be a $\PL(\idis)$-context and let $\psi_1$ and $\psi_2$ be $\PL(\idis)$-formulas. Assume that $\gamma_1\dfn\varphi(\psi_1 \,/\, *)$ is valid. We will show that then $\gamma_2\dfn\varphi\big((\psi_1\idis \psi_2) \,/\, *\big)$ is valid. Let $F$ and $G$ be the sets of exactly all $\idis$-selection functions for $\gamma_1$ and $\gamma_2$, respectively. By Proposition \ref{normalform}
\[
\gamma_1 \equiv {\Idis}_{f\in F} \gamma_1^f \,\text{ and }\, \gamma_2 \equiv {\Idis}_{g\in G} \gamma_2^g.
\]
Since $\gamma_1$ is valid, it follows, by Proposition \ref{disjunctionprop}, that $\gamma_1^{f'}$ is valid for some $f'\in F$. Since clearly, for every $f\in F$, there exists some $g\in G$ such that $\gamma_1^f=\gamma_2^g$, it follows that there exists some $g'\in G$ such that $\gamma_2^{g'}$ is valid. Thus $\gamma_2$ is valid.

In order to prove completeness, assume that a $\PL(\varovee)$-formula $\varphi$ is valid. Let $F$ be the set of exactly all $\idis$-selection functions for $\varphi$. By Propositions \ref{normalform} and \ref{disjunctionprop}, there exists an $f\in F$ such that the $\PL$ formula $\varphi^f$ is valid. Since $\mathbf{H}_{\PL}$ is complete, $\varphi^f$ is provable also in $\mathbf{H}_{ \PL(\varovee)}$. Clearly by using the rules $(I\idis 1)$ and $(I\idis 2)$ repetitively to $\varphi^f$, we eventually obtain $\varphi$. Thus we conclude that $\mathbf{H}_{\PL(\idis)}$ is complete.
\end{proof}
%
%
\subsection{Axiomatizations for $\PD$, $\MDL$, and $\EMDL$}
The following equivalence was observed by V\"a\"an\"anen in \cite{va08}:
\begin{equation}\label{translation1}
\dep{p_1,\dots,p_n,q} \equiv \bigvee_{a_1,\dots,a_n\in \{\bot,\top\}}
\bigwedge \big\{  p_1^{a_1}, \ldots, p_n^{a_n}, (q\varovee q^{\bot}) \big \}.
\end{equation}
Ebbing et all. (\cite{EHMMVV13}) extended this observation of V\"a\"an\"anen into the following equivalence concerning $\EMDL$:
\begin{equation}\label{translation2}
\dep{\varphi_1,\dots,\varphi_n,\psi} \equiv \bigvee_{a_1,\dots,a_n\in \{\bot,\top\}}
\bigwedge \big\{  \varphi_1^{a_1}, \ldots, \varphi_n^{a_n}, (\psi\varovee \psi^{\bot}) \big \}.
\end{equation}
These equivalences demonstrate the existence of compositional translations from $\PD$ into $\PL(\idis)$, and from $\MDL$ and $\EMDL$ into $\ML(\idis)$, respectively. 

We will use the insight that rises from combining the above equivalences with the Propositions \ref{disjunctionprop} and \ref{normalform} in order to construct axiomatizations for $\PD$, $\MDL$, and $\EMDL$, respectively.
Recall that when $\vec{a}$ is a finite tuple of symbols, we use $a_j$ to denote the $j$th member of $\vec{a}$.
For each natural number $n\in\mathbb{N}$, and function $f:\{\bot,\top\}^n\to \{\top,\bot\}$ we have the following rules:
\begin{small}
\begin{prooftree}
\AxiomC{$\varphi\Big(\bigvee_{\vec{a}\in \{\bot,\top\}^n}
\bigwedge \big\{ p_1^{a_1}, \ldots, p_n^{a_n}, q^{f(\vec{a})} \big\} \,/\, * \Big) $} 
\RightLabel{$\big(\PL \,\dep f\big)$}
\UnaryInfC{$\varphi\big(\dep{p_1,\dots,p_n,q} \,/\, * \big)$}
\end{prooftree}
\begin{prooftree}
\AxiomC{$\varphi\Big(\bigvee_{\vec{a}\in \{\bot,\top\}^n} \bigwedge \big\{ \varphi_1^{a_1}, \ldots, \varphi_n^{a_n}, \psi^{f(\vec{a})} \big\} \,/\, * \Big)$}
\RightLabel{$\big(\ML \,\dep f \big)\dagger$}
\UnaryInfC{$\varphi\big(\dep{\varphi_1,\dots,\varphi_n,\psi} \,/\, * \big)$}
\end{prooftree}
\end{small}
where $\dagger$ means that $\varphi_1, \dots, \varphi_n,\psi$ are required to be modal formulae. 
Define $\mathsf{\PL dep} \dfn \{ \big(\PL \,\dep f\big) \mid f:\{\bot,\top\}^n\to \{\top,\bot\} \text{, where }  n\in\mathbb{N}\}$ and $\mathsf{\ML dep} \dfn \{ \big(\ML \,\dep f\big) \mid f:\{\bot,\top\}^n\to \{\top,\bot\} \text{, where } n\in\mathbb{N}\}$. Let
$\mathbf{H}_{\PD}$ and $\mathbf{H}_{\MDL}$ be the extensions of the calculi $\mathbf{H}_{\PL}$ and $\mathbf{H}_{\ML}$ by the rules of $\mathsf{\PL dep}$, respectively. 
Let $\mathbf{H}_{\EMDL}$ be the extension of $\mathbf{H}_{\ML}$ by the rules of $\mathsf{\ML dep}$.

The proof of the following theorem is analogous to that of Theorem \ref{easycalculi}.
\begin{theorem}
Let $\mathcal{L} \in \{\PD, \MDL, \EMDL \}$, $\mathbf{H}_{\mathcal{L}}$ is sound and complete.
\end{theorem}


\section{Labeled tableaus for propositional dependence logics}\label{axiomatization2}
The calculi presented in Section \ref{axiomatization1} have a few clear shortcomings. Foremost, the calculi miss the team semantic nature of these logics. Thus the calculi are in some parts quite complicated. Especially this is the case for the rules $\mathsf{\PL \,dep}$ and $\mathsf{\ML \,dep}$. This seems to be the case also for any concrete implementations of the axiomatizations $\mathbf{H}_{\PL}$ and $\mathbf{H}_{\ML}$ of the negation normal form fragments of $\PL$ and $\ML$, respectively.

In this section we give axiomatizations for $\PD$, $\MDL$, and $\EMDL$ that do not have the shortcomings of the calculi of Section \ref{axiomatization1}. The proof rules of the labeled tableau calculi that we give in this section have a natural and simple correspondence with the truth definitions of connectives and modalities in team semantics.

\subsection{Checking validity via small teams}

The following result (observed e.g. in \cite{virtema14}) follows directly from the fact that $\PL(\idis)$ and $\PD$ are downwards closed, i.e., Proposition \ref{dcprop}.
\begin{proposition}\label{coherencepl}
Let $\varphi$ be a formula of $\PL(\idis)$ or $\PD$ and let $D$ be the set of proposition symbols occurring in $\varphi$. Now
\(
\varphi \text{ is valid} \,\text{ iff }\, \{0,1\}^D\models \varphi.
\)
\end{proposition}

Adapting a notion that
was introduced by Jarmo Kontinen in \cite{jarmo} for first-order dependence logic,
we say that a $\ML(\idis)$- or $\EMDL$-formula $\varphi$ is \emph{$n$-coherent} if the condition 
$$
	\mathrm{K},T\models\varphi\quad\Leftrightarrow\quad \mathrm{K},T'\models\varphi\text{ for all $T'\subseteq T$ 
	such that }|T'|\le n
$$
holds for all Kripke models $\mathrm{K}$ and teams $T$ of $\mathrm{K}$.
The following result for $\ML(\idis)$ was shown in \cite{HeLuSaVi14}. The result for $\EMDL$ follows from the result for $\ML(\idis)$ essentially via the following equivalence.
\begin{equation*}
\dep{\varphi_1,\dots,\varphi_n,\psi} \equiv \bigvee_{a_1,\dots,a_n\in \{\bot,\top\}}
\bigwedge \big\{  \varphi_1^{a_1}, \ldots, \varphi_n^{a_n}, (\psi\varovee \psi^{\bot}) \big \}.
\end{equation*}
For $\varphi\in \ML(\idis)$, we define $\vrank{\varphi}$ to be the number of intuitionistic disjunctions in $\varphi$. For $\psi\in \EMDL$, we define $\vrank{\psi}$ to be the number of intuitionistic disjunctions in the $\ML(\idis)$ formula obtained by using the above equivalence.
\begin{theorem}\label{coherenceml}
Every formula $\varphi$ of $\ML(\idis)$ or $\EMDL$ is $2^{\vrank{\varphi}}$-coherent.
\end{theorem}
The following result follows directly from Theorem \ref{coherenceml}. 
\begin{corollary}\label{cor:coherence}
Let $\varphi$ be a formula of $\ML(\idis)$ or $\EMDL$. The following holds:
\begin{align*}
\varphi \text{ is valid}  \quad\text{iff}\quad &\mathrm{K},T\models \varphi \text{ for every Kripke model $\mathrm{K}$ and every team $T$ of $\mathrm{K}$}\\
&\text{such that $\lvert T \rvert \leq 2^{\vrank{\varphi}}$.}\\
%
\end{align*}
\end{corollary}

\subsection{Tableau Calculi for $\PL$, $\PL(\idis)$, and $\PD$}\label{tableauforpl}
We will now present labeled tableau calculi for $\PL$, $\PL(\idis)$, and $\PD$. In Section \ref{tableauforml} we will extend these calculi to deal with $\ML$, $\MDL$, and $\EMDL$.

Any finite, possibly empty, subset $\alpha \subseteq \mathbb{N}$ is called a {\em label}. We mainly use symbols $\alpha, \beta, \alpha_1,\alpha_2,\beta_1,\beta_2$, etc, in order to refer to labels and symbols $i,j,i_1,i_2,j_1,j_2$, etc,  in order to refer to natural numbers. 
Our tableau calculi are labeled, meaning that the formulas occurring in the tableau rules are \emph{labeled formulae}, i.e., of the form $\alpha:\varphi$, where $\alpha$ a label and $\varphi$ is a formula of some logic $\mathcal{L}$. Labels correspond to teams and the elements of labels, i.e., natural numbers, correspond to points in a model. The intended reading of the labeled formula $\alpha:\varphi$ is that $\alpha$ denotes some team that {\em falsifies} $\varphi$. A tableau in these calculi is just a well-founded, finitely branching tree in which each node is labeled by a labeled formula, and the edges represent applications of the tableau rules. The tableau rules needed for axiomatizing $\PL$, $\PL(\idis)$, and $\PD$ are given in Table \ref{table:tab_rule_prop}.

\begin{table}[btp]
\begin{center}
\begin{footnotesize}
\AxiomC{$\{i_1,\dots,i_k\}\,:\,p$}
\RightLabel{$(Prop)$}
\UnaryInfC{$\{i_1\}\, : \, p \, \mid\, \dots\, \mid\, \{i_k\}\, : \, p$}
\DisplayProof
\quad %
\AxiomC{$\{i_1,\dots,i_k\}\,:\,\neg p$}
\RightLabel{$(\neg Prop)$}
\UnaryInfC{$\{i_1\}\, : \, \neg p \, \mid\, \dots\, \mid\, \{i_k\}\, : \, \neg p$}
\DisplayProof
\AxiomC{$\alpha\,:\,(\varphi\land\psi)$}
\RightLabel{$(\land)$}
\UnaryInfC{$\alpha\,:\,\varphi \,\mid\, \alpha\,:\,\psi$}
\DisplayProof
\quad
\AxiomC{$\alpha \,:\,(\varphi\lor\psi)$}
\RightLabel{$(\lor)$ where $\beta\subseteq \alpha$}
\UnaryInfC{$\beta \,:\,\varphi \,\mid\, \alpha\setminus \beta \,:\, \psi$}
\DisplayProof
\quad
\AxiomC{$\alpha \,:\,(\varphi\idis\psi)$}
\RightLabel{$(\idis)$}
\UnaryInfC{$\alpha \,:\,\varphi$}
\noLine
\UnaryInfC{$\alpha \,:\,\psi$}
\DisplayProof
\begin{prooftree}
\AxiomC{$\alpha \,:\, \dep{p_1,\dots,p_n,q}$}
\RightLabel{$(Split)\dagger$}
\UnaryInfC{$\alpha_1 \,:\, \dep{p_1,\dots,p_n,q} \,\mid\, \dots \,\mid\, \alpha_k \,:\, \dep{p_1,\dots,p_n,q}$}
\end{prooftree}
$\dagger$: $\alpha_1, \dots, \alpha_k$ are exactly all subsets of $\alpha$ of cardinality $2$.
\begin{prooftree}
\AxiomC{$\{i_1,i_2\} \,:\, \dep{p_1,\dots,p_n,q}$}
\RightLabel{$(\PL \, dep)$$\ddagger$}
\UnaryInfC{$\{i_1\} \,:\, p_1^{g_{1}(1)} \,\mid\, \dots \,\mid\, \{i_1\} \,:\, p_1^{g_{k}(1)}$}
\alwaysNoLine
\UnaryInfC{$\{i_2\} \,:\, p_1^{g_{1}(1)} \,\mid\, \dots \,\mid\, \{i_2\} \,:\, p_1^{g_{k}(1)}$}
\UnaryInfC{$ \vdots\quad\quad\quad\quad\quad \vdots  \quad\quad\quad\quad\quad\vdots$}
\UnaryInfC{$\{i_1\} \,:\, p_n^{g_{1}(n)} \,\mid\, \dots \,\mid\, \{i_1\} \,:\, p_n^{g_{k}(n)}$}
\UnaryInfC{$\{i_2\} \,:\, p_n^{g_{1}(n)} \,\mid\, \dots \,\mid\, \{i_2\} \,:\, p_n^{g_{k}(n)}$}
\UnaryInfC{$\{i_1,i_2\} \,:\, q \,\mid\, \dots \,\mid\, \{i_1,i_2\} \,:\, q$}
\UnaryInfC{$\{i_1,i_2\} \,:\, \neg q \,\mid\, \dots \,\mid\, \{i_1,i_2\} \,:\, \neg q$}
\end{prooftree}
$\ddagger$: $g_1, \dots g_k$ are exactly all functions with domain $\{1, \ldots, n\}$ and co-domain $\{ \top, \bot \}$.
\end{footnotesize}
\end{center}
\caption{Tableau Rules for $\mathbf{T}_\PL$, $\mathbf{T}_{\PL(\idis)}$, and $\mathbf{T}_\PD$}
\label{table:tab_rule_prop}
\end{table}

In the construction of tableaus, we impose a rule that a labeled formula is never added to a tableau branch in which it already occurs. A \emph{saturated tableau} is a tableau in which no rules can be applied or the application of the rules have no effect on the tableau. A \emph{saturated branch} is a branch of a saturated tableau. A branch of a tableau is called \emph{closed} if it contains at least one of the following:
\begin{enumerate}
\item Both $\alpha:p$ and $\alpha:\neg p$, for some label $\alpha$ and proposition symbol $p$.
\item $\emptyset:\varphi$, for some formula $\varphi$.
\item $\{i\}: \dep{p_1,\dots,p_n,q}$, for some proposition $p_1,\dots,p_n,q$, $n\in\mathbb{N}$, and $i\in\mathbb{N}$.
\end{enumerate}
If a branch of a tableau is not closed it is called \emph{open}. A tableau is called \emph{closed} if every branch of the tableau is closed. A tableau is called \emph{open} if at least one branch in the tableau is open.

Let $\mathbf{T}_{\PL}$ denote the calculi consisting of the rules $(Prop)$, $(\neg Prop)$, $(\land)$, and $(\lor)$ of Table \ref{table:tab_rule_prop}.
Let $\mathbf{T}_{\PL(\idis)}$ denote the extension of $\mathbf{T}_{\PL}$ by the rule $(\idis)$ of Table \ref{table:tab_rule_prop}, and
$\mathbf{T}_{\PD}$ denote the extension of $\mathbf{T}_{\PL}$ by the rules $(Split)$ and $(\PL \, dep)$ of Table \ref{table:tab_rule_prop}.

Let $\varphi$ be a formula of $\mathcal{L}\in \{\PL,\PL(\idis),\PD\}$. We say that a tableau $\mathcal{T}$ is a \emph{tableau for $\varphi$} if the root of $\mathcal{T}$ is $\{1,\dots,2^{\vrank{\varphi}}\}:\varphi$ and $\mathcal{T}$ is obtained by applying the rules of $\mathbf{T}_{\mathcal{L}}$. We say that $\varphi$ is \emph{provable} in $\mathbf{T}_{\mathcal{L}}$ and write $\vdash_{\mathbf{T}_{\mathcal{L}}}\varphi$ if there exists a closed tableau for $\varphi$.

\begin{theorem}[Termination of $\mathbf{T}_{\PL}$, $\mathbf{T}_{\PL(\idis)}$, and $\mathbf{T}_{\PD}$]\label{termination}
Let $\mathcal{L}$ be a logic in $\{\PL,\PL(\idis),\PD\}$ and $\varphi$ a $\mathcal{L}$-formula. Every tableau for $\varphi$ in $\mathbf{T}_{\mathcal{L}}$ is finite.
\end{theorem}

\begin{proof}
Let $\mathcal{T}$ be a tableau for $\varphi$. By definition, the root of $\mathcal{T}$ is $\alpha:\varphi$, for some finite $\alpha$. Clearly every application of the tableau rules either decreases the size of the label or the length of the formula. Note also that the rule $(\lor)$ can be applied to any $\beta:\psi\in \mathcal{T}$ only finitely many times. Thus $\mathcal{T}$ must be finite.
\end{proof}
\begin{lemma}\label{branchtomodel}
If there exists a saturated open branch for $\varphi$ then $\varphi$ is not valid. 
\end{lemma}
\begin{proof}
Let $\mathcal{B}$ be a saturated open branch for $\varphi$ and let $\Phi$ be the set of proposition symbols that occur in $\varphi$. Let $\alpha : \varphi$ denote the root of the branch $\mathcal{B}$. It is easy to check that if $\beta:\psi$ is a labeled formula in $\mathcal{B}$ then $\beta\subseteq \alpha$. For each $i \in \alpha$ we define an assignment $s_i:\Phi \to \{0,1\}$ such that
\[
s_i(p)\dfn
\begin{cases}
1 &\text{if the labeled formula $\{i\}:\neg p$ occurs in the branch $\mathcal{B}$,}\\
0 &\text{otherwise.}
\end{cases}
\]
It is easy to show by induction that if a labeled formula $\beta:\psi$ occurs in the branch $\mathcal{B}$ then $X_\beta\not\models \psi$, where $X_\beta=\{s_i \mid i\in \beta \}$. Thus $\varphi$ is not valid.
\end{proof}
\begin{theorem}[Completeness of $\mathbf{T}_{\PL}$, $\mathbf{T}_{\PL(\idis)}$, and $\mathbf{T}_{\PD}$]
Let $\mathcal{L}$ be any of the logics in $\{\PL,\PL(\idis),\PD\}$. The calculus $\mathbf{T}_{\mathcal{L}}$ is complete. 
\end{theorem}
\begin{proof}
Fix $\mathcal{L}\in\{\PL,\PL(\idis),\PD\}$.
Assume $\not\vdash_{\mathbf{T}_{\mathcal{L}}} \varphi$. Thus every tableau for $\varphi$ is open. From Theorem \ref{termination} it follows that there exists a saturated open tableau for $\varphi$. Thus there exists a saturated open branch for $\varphi$. Thus, by Lemma \ref{branchtomodel}, $\not\models_{\mathcal{L}}\varphi$.
\end{proof}

\input{plsound.tex}

\begin{theorem}[Soundness of $\mathbf{T}_{\PL}$, $\mathbf{T}_{\PL(\idis)}$, and $\mathbf{T}_{\PD}$]
Let $\mathcal{L}$ be any of the logics in $\{\PL,\PL(\idis),\PD\}$. The calculus $\mathbf{T}_{\mathcal{L}}$ is sound. 
\end{theorem}
\begin{proof}
Fix $\mathcal{L}\in\{\PL,\PL(\idis),\PD\}$.
Assume that $\not\models_{\mathcal{L}} \varphi$. By Lemma \ref{modeltobranch}, there is an open saturated branch in every saturated tableau of $\varphi$ in $\mathbf{T}_{\mathcal{L}}$. Therefore, and since, by Theorem \ref{termination}, every tableau of $\varphi$ in $\mathbf{T}_{\mathcal{L}}$ is finite, there does not exists any closed tableau for $\varphi$ in $\mathbf{T}_{\mathcal{L}}$. Thus $\not\vdash_{\mathbf{T}_{\mathcal{L}}}\varphi$.
\end{proof}

\subsection{Tableau Calculi for $\ML$, $\ML(\idis)$, $\MDL$, and $\EMDL$}\label{tableauforml}
%
In addition to labeled formulas, the tableau rules for modal logics contain \emph{accessibility formulas} of the form $i\mathsf{R}j$, where $i,j\in\mathbb{N}$. The intended interpretation of $i\mathsf{R}j$ is that the point denoted by $j$ is accessible by the relation $R$ from the point denoted by $i$. The tableau rules for the calculi are given in Tables 1 and 2.

In the construction of tableaus, in addition to the rules given in Section \ref{tableauforpl}, we impose that the tableau rule $(\Box)$ is never applied twice to the same labeled formula in any branch. The definitions of open, closed and saturated tableau and branch are as in Section \ref{tableauforpl} with the following addition rule: A branch is called \emph{closed} also if it contains a labeled rule $\{i\}:\dep{\varphi_1,\dots,\varphi_n,\psi}$, for some $i\in\mathbb{N}$, $\varphi_1,\dots,\varphi_n,\psi\in \ML$ and $n\in\mathbb{N}$.

Let $\mathbf{T}_{\ML}$, $\mathbf{T}_{\ML(\idis)}$, and $\mathbf{T}_{\MDL}$ denote the extensions of $\mathbf{T}_{\PL}$, $\mathbf{T}_{\PL(\idis)}$, and $\mathbf{T}_{\PD}$ by the rules $(\Diamond)$ and $(\Box)$ of Table \ref{table:tab_rule_ml}, respectively. 
Let $\mathbf{T}_{\EMDL}$ denote the extension of $\mathbf{T}_{\ML}$ by the rules $(Split)$ of Table \ref{table:tab_rule_prop} and $(\ML \, dep)$ of Table \ref{table:tab_rule_ml}.

Let $\varphi$ be a formula of $\mathcal{L}\in\{\ML,\ML(\idis),\MDL,\EMDL\}$. We say that a tableau $\mathcal{T}$ is a \emph{tableau for $\varphi$} if the root of $\mathcal{T}$ is $\{1,\dots, 2^{\vrank{\varphi}}\}:\varphi$ and $\mathcal{T}$ is obtained by applying the rules of $\mathbf{T}_{\mathcal{L}}$. We say that $\varphi$ is \emph{provable} in $\mathbf{T}_{\mathcal{L}}$ and write $\vdash_{\mathbf{T}_{\mathcal{L}}}\varphi$ if there exists a closed tableau for $\varphi$.

\begin{table}[tbp]
\begin{footnotesize}
\begin{center}
\alwaysNoLine
\AxiomC{$i_1 \mathsf{R} j_1$}
\UnaryInfC{$\vdots$}
\UnaryInfC{$i_n \mathsf{R} j_n$}
\UnaryInfC{$\{i_1,\dots,i_n\} \,:\,\Diamond \varphi$}
\alwaysSingleLine
\RightLabel{$(\Diamond)$}
\UnaryInfC{$\{j_1,\dots,j_n\} \,:\,\varphi$}
\DisplayProof
%
\quad
\AxiomC{$\alpha \,:\,\Box \varphi$}
\RightLabel{$(\Box)\dagger$}
\UnaryInfC{$ f_1(1) \mathsf{R} i_1 \,\mid\, \dots \,\mid\, f_k(1) \mathsf{R} i_1$}
\alwaysNoLine
\UnaryInfC{$ \vdots\quad\quad\quad \vdots  \quad\quad\quad\vdots$}
\UnaryInfC{$ f_1(t) \mathsf{R} i_t \,\mid\, \dots \,\mid\, f_k(t) \mathsf{R} i_t$}
\UnaryInfC{$\{i_1,\dots i_t\} \,:\, \varphi \,\mid\, \dots \,\mid\, \{i_1,\dots i_t\} \,:\, \varphi$}
\DisplayProof
\begin{prooftree}
\AxiomC{$\{i_1,i_2\} \,:\, \dep{\varphi_1,\dots,\varphi_n,\psi}$}
\RightLabel{$(\ML \, dep) \ddagger$}
\UnaryInfC{$\{i_1\} \,:\,\varphi_1^{h_1(1)} \,\mid\, \dots \,\mid\, \{i_1\} \,:\, \varphi_1^{h_k(1)}$}
\alwaysNoLine
\UnaryInfC{$\{i_2\} \,:\, \varphi_1^{h_1(1)} \,\mid\, \dots \,\mid\, \{i_2\} \,:\, \varphi_1^{h_k(1)}$}
\UnaryInfC{$ \vdots\quad\quad\quad\quad\quad \vdots  \quad\quad\quad\quad\quad\vdots$}
\UnaryInfC{$\{i_1\} \,:\, \varphi_n^{h_1(n)} \,\mid\, \dots \,\mid\, \{i_1\} \,:\, \varphi_n^{h_k(n)}$}
\UnaryInfC{$\{i_2\} \,:\, \varphi_n^{h_1(n)} \,\mid\, \dots \,\mid\, \{i_2\} \,:\, \varphi_n^{h_1(n)}$}
\UnaryInfC{$\{i_1,i_2\} \,:\, \psi \,\mid\, \dots \,\mid\, \{i_1,i_2\} \,:\, \psi$}
\UnaryInfC{$\{i_1,i_2\} \,:\, \psi^\bot \,\mid\, \dots \,\mid\, \{i_1,i_2\} \,:\, \psi^\bot$}
\end{prooftree}
$\dagger$: $t=2^{\vrank{\varphi}}$ and $f_1,\dots,f_k$ denote exactly all functions with domain $\{1,\dots, t\}$ and co-domain $\alpha$, and $i_1,\dots,i_t$ are fresh and distinct.\\
$\ddagger$: $h_1, \dots h_k$ denotes all the functions with domain $\{1,\dots, n\}$ and co-domain $\{ \top, \bot \}$. 
\end{center}
\end{footnotesize}
\caption{Additional Tableau Rules for $\mathbf{T}_\ML$, $\mathbf{T}_{\ML(\idis)}$, $\mathbf{T}_\MDL$ and $\mathbf{T}_\EMDL$ }
\label{table:tab_rule_ml}
\end{table}

\begin{definition}
Let $\mathcal{B}$ be a branch of a tableau and let $\alpha : \varphi$ be the root of $\mathcal{B}$. Recall that $\mathrm{Index}(\mathcal{B})$ denotes the set of all natural numbers that occur in $\mathcal{B}$. For $i,j\in\mathrm{Index}(\mathcal{B})$, we write $i \prec_{\mathcal{B}} j$ if $i R j$ occurs in $\mathcal{B}$. By $\prec_{\mathcal{B}}^{\ast}$ and $\preceq_{\mathcal{B}}^{\ast}$, 
we mean the transitive closure and the reflexive and transitive closure of $\prec_{\mathcal{B}}$, respectively. 
Moreover, we define
\begin{align*}
&\mathrm{Level}_{\mathcal{B}}(i) \dfn \lvert\{j\in\mathrm{Index}(\mathcal{B}) \mid i_0 \prec_{\mathcal{B}}^{\ast} j \preceq_{\mathcal{B}}^{\ast} i, \text{ for some $i_0\in \alpha$}\} \rvert,\\
&\mathrm{Layer}_{\mathcal{B}}(n) \dfn \{j\in \mathrm{Index}(\mathcal{B}) \mid  \mathrm{Level}_{\mathcal{B}}(j)=n\}.
\end{align*}
\end{definition}
It is easy to see that, for every branch $\mathcal{B}$, the graph $(\mathrm{Index}(\mathcal{B}),\prec_{\mathcal{B}})$ is a well-founded forest.
%
\begin{theorem}[Termination of $\mathbf{T}_{\ML}$, $\mathbf{T}_{\ML(\idis)}$, $\mathbf{T}_{\MDL}$, and $\mathbf{T}_{\EMDL}$]\label{terminationml}
Let $\varphi$ be a formula of $\ML$, $\ML(\idis)$, $\MDL$, or $\EMDL$. Every tableau for $\varphi$ is finite.
\end{theorem}
\begin{proof}
Let $\mathcal{T}$ be a tableau for $\varphi$ and let $\alpha : \varphi$ denote the root of $\mathcal{T}$. By definition $\alpha$ is finite. Clearly, by the definitions of the tableau rules, if $\beta : \psi$ occurs in $\mathcal{T}$ then $\lvert\beta\rvert \leq \lvert \alpha \rvert$. From this and from the definitions of the tableau rules, it is easy to see that $\mathcal{T}$ is a finitely branching tree. Thus from K\"onig's lemma it follows that $\mathcal{T}$ is infinite if and only if $\mathcal{T}$ has an infinite branch.
%
%

Let $\mathcal{B}$ be an arbitrary branch of $\mathcal{T}$. We will show that $\mathcal{B}$ is finite.

\vspace{1mm}
\noindent \textit{Claim 1} If $\alpha : \varphi$ occurs in $\mathcal{B}$ then, for every $i,j\in\alpha$, $\mathrm{Level}_{\mathcal{B}}(i)=\mathrm{Level}_{\mathcal{B}}(j)$.

\vspace{1mm}
\noindent \textit{Claim 2} For each $k\in\mathbb{N}$ the set $\mathrm{Layer}_{\mathcal{B}}(k)$ is finite.

\vspace{1mm}
\noindent \textit{Claim 3} There is a $k\in\mathbb{N}$ such that $\mathrm{Layer}_{\mathcal{B}}(k)=\emptyset$.

\vspace{1mm}
Note first that if $\mathrm{Layer}_{\mathcal{B}}(k)=\emptyset$ then $\mathrm{Layer}_{\mathcal{B}}(n)=\emptyset$, for every $n\geq k$. Thus from Claims 2 and 3 it follows that only finitely many labels occur in $\mathcal{B}$. Note also that, for every labeled formula $\beta : \psi$ that occurs in $\mathcal{B}$, $\psi$ is either a subformula of $\varphi$ or a subformula of some $\theta^\bot$, where $\theta$ is an $\ML$ subformula of $\varphi$. Thus only finitely many formulas occur in $\mathcal{B}$. Thus $\mathcal{B}$ is finite.

Proof of Claim 1 is easy. We will sketch the proofs of Claims 2 and 3.

\vspace{1mm}
\noindent \textit{Proof sketch of Claim 2}.
Claim 2 follows from Claim 1 by induction: Clearly $\mathrm{Layer}_{\mathcal{B}}(0)$ is finite. $\mathrm{Layer}_{\mathcal{B}}(k+1)$ is generated via applications of the tableau rule $(\Box)$ to labeled formulae $\beta:\Box \psi$ of the branch $\mathcal{B}$, where $\beta\subseteq \mathrm{Layer}_{\mathcal{B}}(k)$ and $\Box \psi$ is either a subformula of $\varphi$ or a subformula of some $\theta^\bot$, where $\theta$ is an $\ML$ subformula of $\varphi$.

\vspace{1mm}
\noindent \textit{Proof sketch of Claim 3}.
%
For finite labels $\beta$, define
\[
m_\mathcal{B}(\beta) \dfn \mathrm{max}\{\lvert \varphi \rvert \mid \beta_1 : \varphi \text{ occurs in $\mathcal{B}$ and $\beta_1\cap\beta\neq \emptyset$}\}.
\]
For finite labels $\beta$, define $M_\mathcal{B}(\beta : \psi) \dfn (m_\mathcal{B}(\beta), \lvert \psi \rvert, \lvert \beta \rvert)$. The ordering between the tuples if defined as follows:
\[
(i,j,k)<(k,l,m) \text{ iff } i<k \text{ or } (i=k \text{ and } j<l) \text{ or } (i=k \text{ and } j=l \text{ and } k<m).
\]
Note that for every labeled formula $\beta : \psi$ that occurs in $\mathcal{B}$ it holds that $m_\mathcal{B}(\beta)< m_\mathcal{B}(\alpha)$, $\lvert \psi\rvert \leq \lvert \varphi \rvert$ and $\lvert \beta\rvert \leq \lvert \alpha \rvert$. Thus the ordering of the tuples is well-founded. Furthermore it is easy to check that application of each tableau rule decreases the measure $M_\mathcal{B}$. For finite collections of labeled formulas $\Gamma$, define $\mathcal{M}_\mathcal{B}(\Gamma) \dfn \mathrm{max}\{M_\mathcal{B}(\beta:\psi) \mid \beta:\psi\in \Gamma \}$. It is straightforward to show that, for every $k\in\mathbb{N}$, either $\mathcal{M}_\mathcal{B}\big(\mathrm{Layer}_{\mathcal{B}}(k+1)\big)<\mathcal{M}_\mathcal{B}\big(\mathrm{Layer}_{\mathcal{B}}(k)\big)$ or $\mathrm{Layer}_{\mathcal{B}}(k+1)=\emptyset$. From this the claim follows. 
\end{proof}

\input{mlsound.tex}

\begin{theorem}[Soundness of $\mathbf{T}_{\ML}$, $\mathbf{T}_{\ML(\idis)}$, $\mathbf{T}{\MDL}$, and $\mathbf{T}{\EMDL}$]
Let $\mathcal{L}\in\{\ML,\ML(\idis),\MDL, \EMDL \}$, the calculus $\mathbf{T}_{\mathcal{L}}$ is sound. 
\end{theorem}

\begin{proof}
Fix $\mathcal{L}\in\{\ML,\ML(\idis),\MDL, \EMDL \}$. 
Assume that $\not\models_{\mathcal{L}} \varphi$. By Lemma \ref{modeltobranchml}, there is an open saturated branch in every saturated tableau of $\varphi$ in $\mathbf{T}_{\mathcal{L}}$. Therefore, and since, by Theorem \ref{terminationml}, every tableau of $\varphi$ in $\mathbf{T}_{\mathcal{L}}$ is finite, there does not exists any closed tableau for $\varphi$ in $\mathbf{T}_{\mathcal{L}}$. Thus $\not\vdash_{\mathbf{T}_{\mathcal{L}}}\varphi$.
\end{proof}

\input{mlcomplete.tex}

\begin{theorem}[Completeness of $\mathbf{T}_{\ML}$, $\mathbf{T}_{\ML(\idis)}$, $\mathbf{T}{\MDL}$, and $\mathbf{T}{\EMDL}$]
Let $\mathcal{L}\in\{\ML,\ML(\idis),\MDL, \EMDL \}$, the calculus $\mathbf{T}_{\mathcal{L}}$ is complete. 
\end{theorem}
\begin{proof}
Fix $\mathcal{L}\in\{\ML,\ML(\idis),\MDL, \EMDL \}$. 
Assume that $\not\vdash_{T_{\mathcal{L}}} \varphi$. Thus every tableau for $\varphi$ is open. From Theorem \ref{terminationml} it follows that there exists a saturated open tableau for $\varphi$. Thus there exists a saturated open branch for $\varphi$. Thus, by Lemma \ref{branchtomodelml}, $\not\models_{\mathcal{L}}\varphi$.
\end{proof}

\section{Conclusion}
We gave sound and complete Hilbert-style axiomatizations for $\PL$, $\PL(\idis)$, $\PD$, $\ML(\idis)$, $\MDL$, and $\EMDL$. In addition, we presented novel labeled tableau calculi for these logics. We proved soundness, completeness and termination for each of the calculi presented.

\bibliographystyle{splncs03}

\label{LastPage}
\end{document}

%% file: plsound.tex
\begin{definition}
Let $\mathcal{B}$ be a tableau branch and $\mathrm{Index}(\mathcal{B})$ the set of all natural numbers occurring in $\mathcal{B}$. We say that  $\mathcal{B}$ is {\em faithful} to a propositional team $X$ by a mapping $f: \mathrm{Index}(\mathcal{B}) \to X$ if, for all $\alpha: \varphi \in \mathcal{B}$, $f[\alpha] \not\models \varphi$. 
\end{definition}

\begin{lemma}
\label{modeltobranch}
Let $\mathcal{L}$ be a logic in $\{ {\PL}, {\PL(\idis)},{\PD}\}$. If $\varphi \in \mathcal{L}$ is not valid then there is an open saturated branch in every saturated tableau of $\varphi$ in $\mathbf{T}_{\mathcal{L}}$. 
\end{lemma}

\begin{proof}
Assume $\not\models_{\mathcal{L}}\varphi$. 
Let $\Phi$ be the set of all proposition symbols occurring in $\varphi$. By Proposition \ref{coherencepl}, $\{0,1\}^{\Phi} \not\models \varphi$. 
Put $\alpha$ := $\{1,\ldots, 2^{|\Phi|} \}$. We fix a bijection $f: \alpha \to \{0,1\}^{\Phi}$. 
Let $\mathcal{T}$ be an arbitrary saturated tableau for $\varphi$. 
By Theorem \ref{termination}, $\mathcal{T}$ is finite and, by definition, the root of $\mathcal{T}$ is $\alpha:\varphi$. 
Note that $\mathrm{Index}(\mathcal{B})$ = $\alpha$, for every branch $\mathcal{B}$ with the root $\alpha:\varphi$. 
We will show that there is an open saturated branch in $\mathcal{T}$. 

First, we establish that $\mathcal{B}_{0}$ := $\{ \alpha: \varphi \}$ is faithful to $\{0,1\}^{\Phi}$ by $f$. 
But, this is easy since $f[\alpha]$ = $\{0,1\}^{\Phi}$. Second, assume that we have constructed a branch $\mathcal{B}_{n}$ such that $\mathcal{B}_{n}$ is faithful to $\{0,1\}^{\Phi}$ by $f$. We will show that at least one extension of $\mathcal{B}_{n}$ by rules of $L$ is faithful to $\{0,1\}^{\Phi}$ by $f$. Here we are concerned with the rule of $(\lor)$ alone. 
Assume that, from $\beta_{1}: (\psi_{1} \lor \psi_{2}) \in \mathcal{B}_{n}$ and the rule of $(\lor)$, we obtain two extensions $\{ \beta_{2}: \psi_{1} \} \cup \mathcal{B}_{n}$ and $\{ \beta_{1} \setminus \beta_{2}: \psi_{2} \} \cup \mathcal{B}_{n}$ for $\beta_{2} \subseteq \beta_{1}$. We want to show that one of the extensions is faithful to $\{0,1\}^{\Phi}$ by $f$. By assumption, we obtain $f[\beta_{1}] \not\models (\psi_{1} \lor \psi_{2})$. By the semantic clause for $\lor$, $f[\beta_{2}] \not\models \psi_{1}$ or $f[\beta_{1}] \setminus f[\beta_{2}]  \not\models \psi_{2}$. 
Since $f[\beta_{1}] \setminus f[\beta_{2}] \subseteq f[\beta_{1} \setminus\beta_{2}]$, 
it follows from downwards closure that $f[\beta_{2}] \not\models \psi_{1}$ or $f[\beta_{1} \setminus\beta_{2}] \not\models \psi_{2}$. 
This implies that at least one of the two extensions is faithful to $\{0,1\}^{\Phi}$ by $f$. We choose one of the faithful extensions as $\mathcal{B}_{n+1}$. 

Since $\mathcal{T}$ is finite and saturated, $\mathcal{B}_{j}$ is a saturated branch in $\mathcal{T}$ for some $j \in \mathbb{N}$. Moreover, since $\mathcal{B}_{j}$ is faithful to $\{0,1\}^{\Phi}$ by $f$, $\mathcal{B}_{j}$ is open.
\end{proof}

%% file: mlsound.tex
\begin{definition}
Let $\mathcal{B}$ be a tableau branch.
We say that $\mathcal{B}$ is {\em faithful} to a Kripke model $\mathrm{K}$  = $(W,R,V)$ if there exists a mapping $f: \mathrm{Index}(\mathcal{B}) \to W$ such that, 
$\mathrm{K}, f[\alpha] \not\models \varphi$ for all $\alpha: \varphi \in \mathcal{B}$, and 
$f(i)Rf(j)$ holds, for every $i\mathsf{R}j \in \mathcal{B}$. 
\end{definition}

\begin{lemma}
\label{modeltobranchml}
Let $\mathcal{L} \in \{ {\ML}, \ML(\idis), {\MDL}, {\EMDL}\}$. If $\varphi \in \mathcal{L}$ is not valid then there is an open saturated branch in every saturated tableau of $\varphi$ in $\mathbf{T}_{\mathcal{L}}$. 
\end{lemma}

\begin{proof}
In this proof, we focus on $\ML(\idis)$. Assume that $\varphi \in \ML(\idis)$ is not valid. By Corollary \ref{cor:coherence}, there is a Kripke model $\mathrm{K}$ = $(W,R,V)$ and a team $T$ of $\mathrm{K}$ such that $\vert T \rvert \leq 2^{\vrank{\varphi}}$ and $\mathrm{K}, T \not\models \varphi$. Put $\alpha_{0}$ := $\{1,\ldots, 2^{\vrank{\varphi}} \}$. Let $\mathcal{T}$ be an arbitrary saturated tableau for $\varphi$. By Theorem \ref{terminationml}, $\mathcal{T}$ is finite and, by definition, the root of $\mathcal{T}$ is $\alpha_{0}:\varphi$. 
We will show that there is an open branch $\mathcal{B}$ in $\mathcal{T}$. 

We first establish that $\mathcal{B}_{0}$ := $\{ \alpha_{0}: \varphi \}$ is faithful to $\mathrm{K}$. 
Let $f: \alpha_{0} \to W$ be any mapping (note: $W$ is non-empty) such that $f[\alpha_{0}]$ = $T$. Clearly $\mathrm{K}, f[\alpha_{0}] \not\models \varphi$, and thus $\mathcal{B}_{0}$ is faithful to $\mathrm{K}$.
Assume then that we have constructed a branch $\mathcal{B}_{n}$ such that $\mathcal{B}_{n}$ is faithful to $\mathrm{K}$. Thus there is a mapping $g:\mathrm{Index}(\mathcal{B}_n)\to W$ such that, for all $\beta:\psi \in \mathcal{B}_{n}$, $\mathrm{K}, g[\beta] \not\models \psi$, and, for all $i\mathsf{R}j \in \mathcal{B}_{n}$, $g(i)Rg(j)$ holds. We will show that any rule-application to $\mathcal{B}_{n}$ generates at least one faithful extension $\mathcal{B}_{n+1}$ to $\mathrm{K}$. Here we are concerned with the rules of $(\Diamond)$ and  $(\Box)$ alone. 
\begin{itemize}
\item[$(\Diamond)$] Assume that $\{i_1,\dots,i_k\}: \Diamond  \psi, i_1 \mathsf{R} j_1, \ldots, i_k \mathsf{R} j_k \in \mathcal{B}_{n}$. Let $\alpha \dfn \{i_1,\dots,i_k\}$ and $\beta \dfn \{j_1,\dots,j_k\}$.
We obtain from our assumption that $\mathrm{K}, g[\alpha] \not\models \Diamond \psi$ and $g[\alpha][R]g[\beta]$. From the semantics of $\Diamond$ it follows that $\mathrm{K}, g[\beta] \not\models \psi$.
Thus $\mathcal{B}_{n+1}$ := $\mathcal{B}_{n} \cup  \{\beta : \psi \}$ is faithful to $\mathrm{K}$. Clearly $\mathcal{B}_{n+1}$ is an extension of $\mathcal{B}$ by the rule $(\Diamond)$.
%
%
%
%
%
\item[$(\Box)$]
Assume that $\alpha: \Box \psi \in \mathcal{B}_{n}$.
We obtain from our assumption that $\mathrm{K}, g[\alpha] \not\models \Box \psi$. By the semantics of $\Box$, it follows that $\mathrm{K}, R[g[\alpha]] \not\models \psi$.
Now, by Theorem \ref{coherenceml}, there exists a team $S \subseteq R[g[\alpha]]$ such that $0<\lvert S\rvert \leq 2^{\vrank{\psi}}$ and $\mathrm{K},S \not\models \psi$.
Fix such $S  \subseteq R[g[\alpha]]$ and let $u_{1}, \ldots, u_{m}$ be the elements of $S$. Since $S  \subseteq R[g[\alpha]]$ there exists a function $h: \{ 1,\ldots, m \} \to \alpha$ such that $g\big(h(l)\big) R u_{l}$, for each $l\leq m$. Let $h':\{ 1,\ldots, 2^{\vrank{\psi}} \} \to \alpha$ denote the expansion of $h$ defined such that $h'(l) \dfn h(m)$ for $m < l \leq  2^{\vrank{\psi}}$.
We then extend our function $g$ to a mapping $g'$ to cover new fresh indexes $\beta \dfn \{j_1,\dots, j_{2^{\vrank{\psi}}}\}$. We define that $g'(j_l)\dfn u_l$, for $l\leq m$, and $g'(j_l)\dfn u_m$ for $m<l\leq 2^{\vrank{\psi}}$.
By construction, we obtain that $\mathrm{K}, g'[\beta] \not\models \psi$ and $g'(h'(l))Rg'(j_{l})$ for all $1 \leq l \leq 2^{\vrank{\psi}}$. Therefore, together with our assumption, $\mathcal{B}_{n+1}$ := $\mathcal{B}_{n} \cup \{ h'(1)\mathsf{R}j_{1}, \ldots, h'(2^{\vrank{\psi}})\mathsf{R}j_{2^{\vrank{\psi}}}, \beta : \psi \}$ is faithful to $\mathrm{K}$ by $g'$. Clearly $\mathcal{B}_{n+1}$ is an extension of $\mathcal{B}$ by the rule $(\Box)$.
\end{itemize}

Since $\mathcal{T}$ is finite and saturated, $\mathcal{B}_{j}$ is a saturated branch in $\mathcal{T}$ for some $j \in \mathbb{N}$. Moreover, since $\mathcal{B}_{j}$ is faithful to $\mathrm{K}$, $\mathcal{B}_{j}$ is open. 
\end{proof}

%% file: mlcomplete.tex
\begin{lemma}
\label{branchtomodelml}
Let $\mathcal{L}\in\{\ML,\ML(\idis),\MDL, \EMDL \}$. 
If there exists an open saturated branch for $\varphi$ in $\mathbf{T}_{\mathcal{L}}$ then $\varphi$ is not valid. 
\end{lemma}

\begin{proof}
Let $\mathcal{B}$ be an open saturated branch in a tableau $\mathcal{T}$ of $\mathbf{T}_{\mathcal{L}}$ starting with $\{1, \ldots, 2^{\vrank{\varphi}}\} :\varphi$. Define the induced Kripke model $\mathrm{K}_{\mathcal{B}}$ = $(W,R,V)$ from $\mathcal{B}$ as follows: 
$W$ := $\mathrm{Index}(\mathcal{B})$; 
$iRj$ iff $i\mathsf{R}j \in \mathcal{B}$; $V(p)$ := $\left\{ i \,|\,\{ i\} : \neg p \in \mathcal{B} \right\}$ for any $p$ occurring in $\mathcal{B}$, otherwise, $V(p)$ := $\emptyset$. 
It is straightforward to prove by induction on $\chi$ that
\(
\alpha: \chi \in \mathcal{B} \text{ implies }\mathrm{K}_{\mathcal{B}}, \alpha \not\models \chi. 
\)
Since $\{1, \ldots, 2^{\vrank{\varphi}} \}: \varphi \in \mathcal{B}$, it follows that $\mathrm{K}_{\mathcal{B}}, \{1, \ldots, 2^{\vrank{\varphi}} \} \not\models \varphi$. Thus $\varphi$ is not valid. 
\end{proof}

%% file: mmdl.bbl
\begin{thebibliography}{10}
\providecommand{\url}[1]{\texttt{#1}}
\providecommand{\urlprefix}{URL }

\bibitem{EHMMVV13}
Ebbing, J., Hella, L., Meier, A., M{\"u}ller, J.S., Virtema, J., Vollmer, H.:
  Extended modal dependence logic. In: WoLLIC 2013, pp. 126--137 (2013)

\bibitem{ebloya}
Ebbing, J., Lohmann, P., Yang, F.: Model checking for modal intuitionistic
  dependence logic. In: Bezhanishvili, G., L\"obner, S., Marra, V., Richter, F.
  (eds.) Logic, Language, and Computation, Lecture Notes in Computer Science,
  vol. 7758, pp. 231--256. Springer (2013)

\bibitem{HeLuSaVi14}
Hella, L., Luosto, K., Sano, K., Virtema, J.: The expressive power of modal
  dependence logic. In: AiML 2014 (2014)

\bibitem{jarmo}
Kontinen, J.: Coherence and Complexity in Fragments of Dependence Logic. Ph.D.
  thesis, University of Amsterdam (2010)

\bibitem{KMSV14}
Kontinen, J., M\"uller, J.S., Schnoor, H., Vollmer, H.: Modal independence
  logic. In: AiML 2014 (2014)

\bibitem{lohvo13}
Lohmann, P., Vollmer, H.: Complexity results for modal dependence logic. Studia
  Logica  101(2),  343--366 (2013)

\bibitem{vollmer13}
M\"uller, J.S., Vollmer, H.: Model checking for modal dependence logic: An
  approach through post's lattice. In: WoLLIC 2013, pp. 238--250 (2013)

\bibitem{Sevenster:2009}
Sevenster, M.: Model-theoretic and computational properties of modal dependence
  logic. J. Log. Comput.  19(6),  1157--1173 (2009)


\bibitem{virtema14}
Virtema, J.: Complexity of validity for propositional dependence logics. In:
  GandALF 2014 (2014)

\bibitem{va08}
V\"a\"an\"anen, J.: Modal dependence logic. In: Apt, K.R., van Rooij, R. (eds.)
  New Perspectives on Games and Interaction, Texts in Logic and Games, vol.~4,
  pp. 237--254 (2008)

\bibitem{fanthesis}
Yang, F.: On Extensions and Variants of Dependence Logic. Ph.D. thesis,
  University of Helsinki (2014)

\end{thebibliography}
